\documentclass[11pt,reqno]{amsart}

\usepackage{amsfonts}
\usepackage[sans]{dsfont}

\usepackage{amsmath,amssymb,amsthm}
\usepackage[margin=1.0in]{geometry}
\usepackage{enumerate}
\usepackage{color}

\theoremstyle{plain}
\newtheorem{theorem}{Theorem}

\newtheorem{lemma}[theorem]{Lemma}

\theoremstyle{definition}

\numberwithin{theorem}{section}
\numberwithin{equation}{section}

\newcommand{\R}{\mathbb{R}}

\newcommand{\cF}{\mathcal{F}}

\newcommand{\cH}{\mathcal{H}}

\def\CO{{\mathcal O}}

\newcommand{\e}{{\epsilon}}
\newcommand{\al}{{\alpha}}

\newcommand{\G}{{\Gamma}}

\newcommand{\lam}{{\lambda}}           

\newcommand{\om}{\omega}

\def\CO{{\mathcal O}}
\renewcommand{\d}{\mathrm{d}}

\def\<{\langle}
\def\>{\rangle}

\newcommand{\ddim}{\operatorname{d\:\!Im}}
\newcommand{\ddre}{\operatorname{d\:\!Re}}

\newcommand{\ran}{\rangle}
\newcommand{\lan}{\langle}
\newcommand{\ra}{\rightarrow}

\newcommand{\Ran}{\operatorname{Ran}}

\newcommand{\p}{{\partial}}

\newcommand{\bfone}{{\bf 1}}
\newcommand{\one}{\mathbf{1}}

\newcommand{\DETAILS}[1]{}
\newcommand{\hp}{H_{p}}
\newcommand{\chp}{\mathcal{H}_{p}}

\newcounter{foo}
\Roman{foo}
\begin{document}

\title[Comment on the photon number condition]{Comment on  the photon number bound and Rayleigh scattering}

\author[J. Faupin]{J{\'e}r{\'e}my Faupin}
\address[J. Faupin]{Institut de Math{\'e}matiques de Bordeaux \\
UMR-CNRS 5251, Universit{\'e} de Bordeaux 1, 
33405 Talence Cedex, France}
\email{jeremy.faupin@math.u-bordeaux1.fr}
\author[I. M. Sigal]{Israel Michael Sigal} 
\address[I. M. Sigal]{Department of Mathematics \\
University of Toronto, 
Toronto, ON M5S 2E4, Canada}
\email{im.sigal@utoronto.ca}

\begin{abstract}
We discuss  photon number bounds for a system of non-relativistic particles coupled to the quantized electromagnetic field (non-relativistic QED),  below the ionization threshold. 
Such a bound was assumed in the proof of asymptotic completeness for Rayleigh scattering in our paper \cite{FaSig} (Condition (1.20) of Theorem 1.3 in \cite{FaSig}). We show how this assumption can be weakened and verified for a class of hamiltonians.
\end{abstract}

\maketitle

\section{ Introduction}\label{sec:intro}

In this note we discuss photon number bounds for non-relativistic particle systems coupled to quantized electromagnetic or phonon field.    (We use the term photon for both photon and phonon.) Such a bound was first proved by  W. De Roeck and A. Kupiainen in  \cite{DRK} for the spin-boson model and a variant of such a bound was assumed  in our proof of asymptotic completeness below the ionization threshold, i.e. for Rayleigh scattering, in \cite{FaSig}. Specifically, we assumed that   the photon number is bounded uniformly in time (Condition (1.20) of Theorem 1.3 in \cite{FaSig}). In this note we show how this assumption can be weakened   and verified for a class of hamiltonians.

In \cite{FaSig}  we consider  the dynamics generated by the Hamiltonian  (here and in what follows we use, without mentioning it, the notation of the  paper \cite{FaSig})
\begin{equation}\label{H}
H=\hp + H_f + I(g),
\end{equation}
acting on  the state space ${\cH}:=\chp\otimes \cF$. Here, $\cH_{p}$ is the particle state  space,  $\cF$ is the bosonic Fock space based on the one-photon space $L^2(\R^3)$,  $\hp$ is a self-adjoint Hamiltonian acting on $\cH_{p}$, and $H_f:=\d \Gamma( \om )$ (where $\omega(k)=|k|$ is the photon dispersion law and $k$ is the photon wave vector) is the photon Hamiltonian acting on $\cF$.

The operator $I(g)$, acting on ${\cH}$, represents an interaction energy labeled by  a coupling family $g(k)$ of operators acting on $\chp$. It is of the form
\begin{equation}\label{I}
I(g):= \int (g^* (k)\otimes a(k)+g (k)\otimes a^*(k))dk,
\end{equation}
with $a^*(k)$ and  $a(k)$ the creation and annihilation operators acting on $\cF$. The coupling operators $g (k)$ are assumed to satisfy
 \begin{equation}\label{g-est}
\|\eta^{|\al|}\p^\al g (k)\|_{\chp} \lesssim |k|^{\mu-|\al|} \xi (k),\quad |\al| \le 2,
\end{equation}
where $\xi (k)$ is an ultraviolet cutoff (a smooth function decaying sufficiently rapidly at infinity) and $\eta$ is an estimating operator (a bounded, positive operator with unbounded inverse) on $\chp$, satisfying
\begin{equation}\label{eta-bnd}
\|\eta^{-n}f(H)\|\lesssim 1 ,
\end{equation}
for any $n=1, 2$ and  $f \in \mathrm{C}_0^\infty( (-\infty,\Sigma))$, where  $\Sigma$ is the ionization threshold.

The proofs  presented here, as well as - as was mentioned in \cite{FaSig} - those in \cite{FaSig}, can be extended to the minimal coupling model with the standard quantum Hamiltonian (see \cite{BoFaSig} for the notations used)
\begin{equation*}
H = \sum  \limits_{j=1}^n \frac{1}{2 m_j} \big( - i \nabla_{x_j} - g_{j} A_\kappa ( x_j ) \big)^2 + V (x) + H_f .
\end{equation*}

Let $\psi_t= e^{ - i t H }\psi_0$ be the solution of the Schr\"odinger equation $ i\p_t\psi_t =H\psi_t$ with an initial condition  $\psi_0\in  \Ran \, E_{(-\infty, \Sigma)}(H)$. The assumption (1.20) of Theorem 1.3 of  \cite{FaSig}  states that
\begin{itemize}
\item For any $\psi_0 \in  D(N^{1/2})$ and uniformly in $t \in [0 , \infty)$, 
\begin{equation}\label{npb-unif}
\| N^{1/2} \psi_t \| \lesssim \| N^{1/2} \psi_0 \| + \| \psi_0 \| .
\end{equation}
\end{itemize}
It can be weakened to one of the following conditions:
\DETAILS{The assumption, in the previous version that the uniform bound
\begin{equation}\label{npb-unif}
\| N^{1/2} \psi_t \| \lesssim \| N^{1/2} \psi_0 \| + \| \psi_0 \| ,
\end{equation}
holds for any $\psi_0$ in $D( N^{1/2} )$, can now be weakened in two respects: \eqref{npb-unif} }

\smallskip

\begin{itemize}
\item [(i)]
\eqref{npb-unif} holds {\it only} for initial states $\psi_0 \in f (H) D(N^{1/2})$,  {\it with} $f \in \mathrm{C}_0^\infty( ( E_{\mathrm{gs}}, \Sigma) )$, 

\item [(i')] 
 There exists a set $\mathcal{D}$ such that $\mathcal{D} \cap D ( \d\G( \omega^{-1/2}Ê\langle y \rangle \omega^{-1/2} )^{\frac12} )$ is dense in $\mathrm{Ran} \, E_{(-\infty , \Sigma) }(H)$ and, for any $\psi_0 \in \mathcal{D}$, 
\begin{equation}\label{npb-unif2}
\| \d\Gamma( \omega^{-1} )^{\frac12} \psi_t \| \lesssim C( \psi_0 ) ,
\end{equation}
uniformly in $t \in [0 , \infty)$, where $C( \psi_0 )$ is a positive constant depending on $\psi_0$.
\end{itemize}

\smallskip

Condition (i) deals only with states below the ionization threshold, while (i')  does not specify the dense set of $\psi_0$'s  and, as a result,  can be verified for the massless spin-boson model by modifying slightly  the proof of De Roeck and Kupiainen in \cite{DRK}. 
 Hence the asymptotic completeness in this case holds with no implicit conditions.

\smallskip

 To verify \eqref{npb-unif2} for the spin-boson model,  we proceed precisely in the same way as in \cite{DRK}, but using the stronger condition on the decay of correlation functions, 
\begin{equation}\label{correl-decay}
 \int_0^\infty d t  \,   (1+ t )^{\alpha} | h(t) | < \infty, \quad \textrm{with} \quad   h(t) :=  \int_{\R^3} d k \,   e^{- i t | k |}  ( 1 + |k|^{-1} )  | g(k) |^2 ,
 \end{equation}
 for some $\alpha \ge 1$, instead of  Assumption A of \cite{DRK}, and bounding the observable $( 1 + \kappa \d \Gamma( \omega^{-1/2} ) )^2$   instead of $e^{\kappa N} $. Assumption C of \cite{DRK} on initial states has to be replaced in the same manner. Assuming that 
 \eqref{g-est} 
 is satisfied with $\mu>0$ (and $\eta=1$), we see that \eqref{correl-decay} holds with $\alpha = 1 + 2\mu$.

The form of the observable  $e^{\kappa N} $ enters \cite{DRK} through the estimate $\|K_{u, v}\|_\diamond \le \lam^2 C | h(u-v)|$ of the operator $K_{u, v}$ defined in \cite[(3.4)]{DRK} and the standard estimate \cite[(4.36)]{DRK}. Both extend readily to our case (the former, with $h(t)$ given in \eqref{correl-decay}).  Moreover, \cite[(4.36)]{DRK} is used in the proof that pressure vanishes - Eq (4.39) in \cite{DRK} - and the latter also follows from our Proposition A.1. 
(We can also use the observable $ \Gamma( \om^{-\lam} )=\d \Gamma( -\lam\ln \om )$ and analyticity - rather than perturbation - in $\lam$.).

 Now we comment on the modifications needed  in order to prove  the result of Theorem 1.3 of \cite{FaSig}  under the new assumptions. These modifications concern only the proof of  the existence of the Deift-Simon wave operators given in Theorem 5.1 of \cite{FaSig}.  
 \begin{itemize}
\item To prove Theorem 5.1 under Assumption (i), we need minor modifications in the proof, relying on slightly strengthened  Lemma 5.2, by using a new estimate on the growth of the observable $N^2$ (in addition to $N$). 
\item The proof of Theorem 5.1 under Assumption (i') is analogous to the one  for Assumption (i). The only difference is that we do not need to introduce an artificial cutoff in the number operator. Instead  we use additional `weighted' propagation estimates, which  are straightforward modifications of the estimates (3.3)--(3.4) in \cite{FaSig}. 
\end{itemize}

In the next two sections we present detailed modifications in our proof in \cite{FaSig}, needed to prove asymptotic completeness for Rayleigh scattering under either Assumption (i) or (i').

We use the notation $\|\psi\|_\rho^2 := \| (\d\G(\omega^{\rho}) +1)^{\frac{1}{2}}\psi_0 \|$  from \cite{FaSig}.

\medskip

\noindent {\bf Acknowledgements.} We are grateful to J\"urg Fr\"ohlich for emphasizing to us the importance of the photon number bounds and to Marcel Griesemer  for criticism of Condition (1.20) of Theorem 1.3 in \cite{FaSig}, which led to this note. The second author's research was supported in part by NSERC under Grant No. NA7901.

\section{Adjustments in Proof of Theorem 5.1 under Condition ($\mathrm{i}$)}\label{sec:pf-ac}

The part of Theorem 5.1 which requires a modification is showing that 
\begin{itemize}
\item the family $W(t):=e^{i\hat{H} t} \check{\Gamma} (j) e^{-iHt}$ form a strong Cauchy sequence as $t\ra\infty$. \end{itemize}
We present here  the corresponding changes. 
Let $\psi_0\in f(H) D( \d\Gamma( \omega^{-1} )^{1/2} )$, $f \in \mathrm{C}_0^\infty( (E_{\mathrm{gs}} , \Sigma ) )$. Lemma \ref{lm:f(H)W2}, proven below, implies that
\begin{align} \label{eq:AC1}
W(t) \psi_0 &= e^{i\hat{H} t} f_1( \hat H ) \check{\Gamma} (j) e^{-i Ht} f_1(H) \psi_0 
+ \mathcal{O}( t^{-\alpha + \frac{1}{2+\mu}}   \|  \psi_0 \|_{-1}), 
\end{align}
where $f_1 \in \mathrm{C}_0^\infty( (E_{\mathrm{gs}} , \Sigma ) )$ is such that $f_1 f = f$. Hence, since our conditions on $\alpha$ imply $\alpha >  1 / ( 2 + \mu )$, it suffices to show that 
\begin{align*}
\widetilde W(t) := e^{i\hat{H} t} f_1( \hat H ) \check{\Gamma} (j) e^{-i Ht} f_1(H)
\end{align*}
form a strong Cauchy sequence as $t\ra\infty$. This is done exactly as in \cite{FaSig} for $W(t)$. It remains to prove the following lemma which is strengthening of the corresponding lemma (Lemma 5.2) of \cite{FaSig}.

The rest of the proof of Theorem 5.1 of \cite{FaSig} under Assumption (i) is exactly the same as in \cite{FaSig}.  \qed

\smallskip
\begin{lemma} \label{lm:f(H)W2}
Assume \eqref{g-est} with $\mu > 0$ and \eqref{eta-bnd}. For any $f \in \mathrm{C}_0^\infty( \Delta )$, $\Delta \subset (E_{\mathrm{gs}} , \Sigma )$, and $\psi_0 \in \Ran  E_\Delta (H) \cap D( \d\Gamma( \omega^{-1} )^{1/2} )$,
\begin{align}
\| ( \check{\Gamma}(j) f(H)  - f( \hat{H} ) \check{\Gamma}(j) ) \psi_t\| \lesssim t^{-\alpha+\frac{1}{2+\mu}} \|Ê\psi_0 \|_{-1}. \label{eq:e0_1}
\end{align}
\end{lemma}
\begin{proof}
Using the Helffer-Sj{\"o}strand formula, we compute $ \check{\Gamma}(j) f(H) \psi_t - f( \hat{H} ) \check{\Gamma}(j) \psi_t=R$, where
\begin{align}
R&:= \frac{1}{\pi} \int \partial_{\bar z} \widetilde{f} (z) ( \hat H - z )^{-1} ( \hat H \check{\Gamma}(j) - \check{\Gamma}(j) H ) ( H - z )^{-1} \psi_t
 \ddre z \ddim z , \label{eq:e2_1}
\end{align}
and $\widetilde f$ is an almost analytic extension of $f$ with the usual properties. We have $\hat H \check{\Gamma}(j) - \check{\Gamma}(j) H = \tilde G_0 -i G_1$, where $\tilde G_0:= U \d {\Gamma} (j, \underline{\omega} j - j \omega)$ and $ G_1:=
 ( I ( g) \otimes \bfone ) \check{\Gamma} (j) - \check{\Gamma} (j) I ( g)$.

We consider $\tilde G_0$. We have $\underline{\omega} j - j \omega =  ( [ \omega , j_0 ] , [ \omega , j_\infty ] )$, and, by Corollary B.3 of Appendix B of \cite{FaSig}, 
\begin{align}\label{eq:s1}
[ \omega , j_\# ] = \frac{ \theta_\e }{ c t^\alpha}  j'_\# + r ,
\end{align}
where $j_\#$ stands for $j_0$ or $j_\infty$, $j'_\#$ is the derivative of $j_\#$ as a function of $\frac{b_\e}{c t^{\al}}$, and $r$ satisfies $\|r \| \lesssim  t^{-2 \alpha+\kappa}$. Since $\theta_\e \le 1$ and since $\kappa < \alpha$, we deduce that $[ \omega , j_\# ] = \CO( t^{-\alpha} )$. By  (C.2) of Appendix C of \cite{FaSig},  
 we then obtain that
\begin{align*}
\| \tilde G_0 (N+1)^{-1} \| = \| (\hat N+1)^{-\frac12} \tilde G_0 (N+1)^{-\frac12} \| \lesssim t^{-\alpha} .
\end{align*}
The equality above follows from $(\hat N + 1 )^{-1/2} \tilde G_0 = \tilde G_0 (N+1)^{-1/2}$. Using, for instance, that $H \in C^1(N)$, we verify that $\|(N+1) ( H-z )^{-1} (N+1)^{-1} \| \lesssim | \mathrm{Im} \, z |^{-2}$, and hence
\begin{align}\label{eq:s3_1}
\|  \tilde G_0 ( H - z )^{-1} \psi_t \| &\lesssim t^{-\alpha} | \mathrm{Im} z |^{-2} \| (N+1) \psi_t \| .
\end{align}

Now, we need the following result, which is a consequence  of the low-momentum bound (A.1) of \cite{FaSig} and whose proof is given below:
Under \eqref{g-est} with $\mu > 0$, we have that 
\begin{equation}\label{lm-bnd_2}
\| N\psi_t\| \lesssim  t^{ \frac{1}{2+\mu} } \| \psi_0 \|_{-1},
\end{equation}
provided $\psi_0 \in f (H) D(\d\G( \omega^{-1} )^{1/2} )$, with $f \in \mathrm{C}_0^\infty( \R )$.  Applying this estimate, we obtain
\begin{align}\label{eq:s3_2}
\|  \tilde G_0 ( H - z )^{-1} \psi_t \| &\lesssim t^{-\alpha+\frac{1}{2+\mu}} | \mathrm{Im} z |^{-2}  \|  \psi_0 \|_{-1}. 
\end{align}

As in  (5.30)--(5.31) of \cite{FaSig}, we have in addition
\begin{align*}
\| G_1 (N+1)^{-\frac12} E_\Delta(H) \| \lesssim t^{-(\mu+\frac32)\alpha} ,
\end{align*}
and hence, using, as above, that $\|(N+1)^{1/2} ( H-z )^{-1} (N+1)^{-1/2} \| \lesssim | \mathrm{Im} \, z |^{-2}$, we obtain
\begin{align}\label{eq:e3b_3}
\| G_1 ( H - z )^{-1} \psi_t \| \lesssim t^{-(\mu+\frac32)\alpha} | \mathrm{Im} z |^{-2}  \|  \psi_0 \|_{N}.
\end{align}

From \eqref{eq:e2_1}, \eqref{eq:s3_2}, \eqref{eq:e3b_3}, the properties of the almost analytic extension $\tilde f$ and the estimate 
$\|( H - z )^{-1}  \| \lesssim  | \mathrm{Im} z |^{-1}$, we conclude that \eqref{eq:e0_1} holds

Finally we prove \eqref{lm-bnd_2}.  By the Cauchy-Schwarz inequality, we have $N^2 \le \d\Gamma( \omega ) \d\Gamma( \omega^{-1} )$, and hence
\begin{align*}
\langle N^2 \rangle_{\psi_t} &\le \langle \d\Gamma( \omega^{-1} )^{\frac12} \d \Gamma( \omega ) \d \Gamma( \omega^{-1} )^{\frac12} \rangle_{\psi_t} \notag \\
&= \langle \d\Gamma( \omega^{-1} )^{\frac12} \d \Gamma( \omega ) (H - E_{\mathrm{gs}} + 1 )^{-1} \d \Gamma( \omega^{-1} )^{\frac12} ( H - E_{\mathrm{gs}} + 1 ) \rangle_{\psi_t} \notag \\
&\quad+ \langle \d\Gamma( \omega^{-1} )^{\frac12} \d \Gamma( \omega ) [ \d \Gamma( \omega^{-1} )^{\frac12} , (H- E_{\mathrm{gs}} + 1)^{-1} ] ( H- E_{\mathrm{gs}} + 1 ) \rangle_{\psi_t} .
\end{align*}
Under Assumption \eqref{g-est} with $\mu > 0$, one verifies that $\d \Gamma( \omega ) [ \d \Gamma( \omega^{-1} )^{\frac12} , (H- E_{\mathrm{gs}} + 1)^{-1} ]$ is bounded. Since $\d \Gamma( \omega ) (H- E_{\mathrm{gs}} + 1)^{-1}$ is also bounded, we obtain
\begin{align}
\langle N^2 \rangle_{\psi_t} \lesssim \| \d\Gamma( \omega^{-1} )^{\frac12} \psi_t \| &\big ( \| \d\Gamma( \omega^{-1} )^{\frac12} ( H- E_{\mathrm{gs}} + 1) \psi_t \|\notag \\& + \| ( H- E_{\mathrm{gs}} + 1) \psi_t \| \big ). \label{eq:AB1}
\end{align}
Applying Proposition A.1 of \cite{FaSig} gives
\begin{align}
\| \d\Gamma( \omega^{-1} )^{\frac12} \psi_t \| \lesssim t^{\frac{1}{2+\mu}} \| \psi_0 \| + \| \d\Gamma( \omega^{-1} )^{\frac12}  \psi_0 \| , \label{eq:AB2}
\end{align}
and
\begin{align}
\| \d\Gamma( \omega^{-1} )^{\frac12} ( H- E_{\mathrm{gs}} + 1) \psi_t \| &\lesssim t^{\frac{1}{2+\mu}} \| \psi_0 \| + \| \d\Gamma( \omega^{-1} )^{\frac12} ( H- E_{\mathrm{gs}} + 1) \psi_0 \| \notag \\
&\lesssim t^{\frac{1}{2+\mu}} \| \psi_0 \| + \| \d \Gamma( \omega^{-1} )^{\frac12} \psi_0 \| , \label{eq:AB3}
\end{align}
where we used in the last inequality that $\d\Gamma( \omega^{-1} )^{\frac12} \tilde f(H) \d\Gamma( \omega^{-1} )^{-\frac12}$ is bounded for any $\tilde f \in \mathrm{C}_0^\infty( \mathbb{R} )$ (this can be verified, for instance, by using that $H \in C^1( \d\Gamma( \omega^{-1} ) )$). Combining \eqref{eq:AB1}, \eqref{eq:AB2} and \eqref{eq:AB3}, we obtain
\eqref{lm-bnd_2} .
This completes the proof of Lemma \ref{lm:f(H)W2}. \end{proof}

\medskip

\section{The proof of the existence of $W_+$ under  Assumption ($\mathrm{i}')$}\label{sec:pfthm5.1-condi'}

The proof of the existence of $W_+$  under Assumption (i')  is similar to the  proof under Assumption (i), except that we do not need to introduce the cutoff $\chi_m$. We use instead  the following weighted propagation estimates, which are  straightforward extensions of the estimates of Theorem  3.1 of \cite{FaSig}: 
\begin{align}\label{mve1'_2}
\int_1^\infty dt\  t^{-\beta}
\| \d {\Gamma} ( \rho_1^*\chi_{\frac{b_\e}{ct^{\beta}} =1}\rho_1Ê)^{\frac{1}{2}} \psi_t \|^2 \lesssim  \| \psi_0 \|^2 ,
\end{align}
for  $\mu$ and $\beta$ as in Theorem 3.1 and any $\psi_0 \in \mathcal{H}$, and, if in addition Assumption (i') holds,  
\begin{align}
& \int_1^\infty dt\  t^{-\beta} \| \d {\Gamma} ( \omega^{-1/2} \chi_{\frac{b_\e}{ct^{\beta}} =1} \omega^{-1/2} )^{\frac{1}{2}} \psi_t \|^2 \lesssim  C( \psi_0 ) ,  \label{mve1_2}
\end{align}
and
\begin{align}
& \int_1^\infty dt\  t^{-\beta} \| \d {\Gamma} ( \rho_{-1}^*\chi_{\frac{b_\e}{ct^{\beta}} =1}\rho_{-1}Ê)^{\frac{1}{2}} \psi_t \|^2 \lesssim  C( \psi_0 ) , \label{mve1'_3}
\end{align}
for any $\psi_0 \in \mathcal{D}$. Here $\rho_\nu:= \chi\theta_\e^{1/2} \omega^{\nu/2}$ (recall that $\chi\equiv \chi_{(\frac{|y|}{\bar c t})^2 \le 1}$). Likewise, under Assumption (i') 
the proof of the maximal velocity estimate of \cite{BoFaSig}, in the form (1.9) of \cite{FaSig}, can easily be extended to the following weighted maximal velocity estimate:
\begin{equation}\label{maxvel-est_2}
\big\Vert \d \Gamma \big( \omega^{-1/2} \chi_{ \vert y \vert  \geq \bar{\mathrm{c}} t  } \omega^{-1/2} \big)^{\frac{1}{2}} \psi_t  \big\Vert \lesssim t^{- \gamma} \big ( \big\Vert ( \d \Gamma ( \omega^{-1/2} \< y \> \omega^{-1/2} ) + 1 )^{\frac12}  \psi_0 \big\Vert + C( \psi_0 ) \big ),
\end{equation}
for any $ \bar{\mathrm{c}}>1$, $\gamma < \min ( \frac{ \mu }{ 2 } \frac{ \bar{\mathrm{c}} - 1 }{ 2 \bar{\mathrm{c}} - 1 } , \frac12 )$ and $\psi_0 \in \mathcal{D} \cap D ( \d\G( \omega^{-1/2}Ê\langle y \rangle \omega^{-1/2} )^{\frac12} )$.

We only mention that to obtain for instance \eqref{mve1_2}, we estimate the interaction term using the estimate (2.11) of \cite{FaSig} with $\delta = - 1/2$ together with Lemma B.6 of Appendix B of \cite{FaSig} and \eqref{npb-unif2}.

 Now, let $\psi_0 \in \mathcal{D} \cap D ( \d\G( \omega^{-1/2}Ê\langle y \rangle \omega^{-1/2} )^{\frac12} )$. We decompose $(\widetilde W(t') - \widetilde W(t))\psi_0$ as in Equations (5.15)--(5.20) of \cite{FaSig}. Using the commutator estimates of Appendix B of \cite{FaSig} and Hardy's inequality, we verify that 
\begin{equation*}
\rho_{-1}^* ( j'_0 , j'_\infty )\rho_{1} = \theta_\e^{1/2}\chi ( j'_0 , j'_\infty )\chi \theta_\e^{1/2} + \CO( t^{-\alpha + (1+\kappa )/2 } ),
\end{equation*}
and likewise for the remainder terms $\mathrm{rem}_t$. Hence Equations  (5.19)--(5.20) of \cite{FaSig} can be transformed into
\begin{align}\label{udj2_2}
&\underline dj = \frac{1}{c t^{\al}} \rho_{1}^* ( j'_0 , j'_\infty )\rho_{-1} +Ê\omega^{1/2} \mathrm{rem}'_t \, \omega^{-1/2} \\ 
&\mathrm{rem}'_t = \mathrm{rem}_t + \CO( t^{ - 2\alpha + (1+\kappa )/2 } ),
\end{align}
where $\mathrm{rem}_t$ is given in (5.20) of \cite{FaSig}. These relations give
\begin{align}\label{tildeG0-deco_2}
G_0 = \widetilde G_0'+  \mathrm{ Rem}'_{t},
\end{align}
where $\widetilde G'_0 := \frac{1}{ct^\alpha} U \d \Gamma ( j ,  \widetilde{ \underline c_t} )$, with  $\widetilde{ \underline c_t} =(\widetilde c_0, \widetilde c_\infty):= ( \rho_{1}^*  j'_0 \rho_{-1} , \rho_{1}^* j'_\infty \rho_{-1} )$, and
\begin{equation*}
\mathrm{Rem}'_t := G_0 - \widetilde G'_0 = U \d \Gamma ( j , \mathrm{rem}'_t ).
\end{equation*}

Next, we consider $\widetilde A = \sup_{\|\hat\phi_0\| = 1}| \int_t^{t'}ds\lan \hat\phi_s,  G_0 \psi_{s}\ran|$, where $\hat\phi_s =e^{-i\hat{H} s} f(\hat H) \hat \phi_0$. Let 
\begin{align*}
& a_{0} = \rho_{1}^*  | j'_0 |^{1/2} , \quad b_{0} = | j'_0 |^{1/2}\rho_{-1}, \\
& a_{\infty } = \rho_{1}^*  | j'_\infty |^{1/2}, \quad b_{\infty } = | j'_\infty |^{1/2}\rho_{-1}.
\end{align*}
We have $\widetilde c_0 = - a_{0} b_{0}$, $\widetilde c_\infty = a_{\infty } b_{\infty }$. 
Exactly as for  (C.1) of Appendix C of \cite{FaSig}, 
one can show that, if $c = ( a_{0} b_{0} ,  a_{\infty } b_{\infty } )$,  where $  a_{0}, b_{0} ,  a_{\infty },  b_{\infty } $ are operators on  $\mathfrak{h}$, then 
 \begin{align} \label{checkG-ineq'_2}
| \lan \hat\phi, \d \check{\Gamma} (j, c) \psi\ran | & \le  \| \d \Gamma( a_{0} a_{0}^* )^{\frac12} \otimes \one \hat \phi \| \| \d {\Gamma} ( b_{0}^* b_{0} )^{\frac12} \psi \| \notag \\
&+ \| \one \otimes \d \Gamma(  a_{\infty }  a_{\infty }^* )^{\frac12} \hat \phi \| \| \d {\Gamma} ( b_{\infty }^* b_{\infty } )^{\frac12} \psi \| .
\end{align} 
Hence  $\widetilde G'_0$ satisfies
\begin{align}\label{G0'-est_2}
|\lan \hat\phi, \widetilde G'_0 \psi\ran | &\le \frac{1}{ct^\alpha} \big (  \| \d \Gamma( a_{0} a_{0}^* )^{\frac12} \otimes \one \hat \phi \| \| \d {\Gamma} ( b_{0}^* b_{0} )^{\frac12} \psi \| \notag \\
&+ \| \one \otimes \d \Gamma(  a_{\infty }  a_{\infty }^* )^{\frac12} \hat \phi \| \| \d {\Gamma} ( b_{\infty }^* b_{\infty })^{\frac12} \psi \| \big ).
\end{align}
By the Cauchy-Schwarz inequality, \eqref{G0'-est_2} implies
\begin{align*}
\int_t^{t'}ds | \lan \hat\phi_s,  \widetilde G'_0 \psi_{s}\ran | & \lesssim \Big ( \int_t^{t'} ds \, s^{-\alpha}
\| \d \Gamma( a_{0} a_{0}^*  )^{\frac12} \otimes \one \hat\phi_s\|^2 \Big )^{\frac12} \Big ( \int_t^{t'} ds \, s^{-\alpha} \| \d {\Gamma} ( b_{0}^* b_{0}  )^{\frac12} \psi_{s}\|^2 \Big )^{\frac12} \notag \\
& + \Big ( \int_t^{t'} ds \, s^{-\alpha} \| \one \otimes \d \Gamma(  a_{\infty }  a_{\infty }^* )^{\frac12} \hat\phi_s\|^2 \Big )^{\frac12} \Big ( \int_t^{t'} ds \, s^{-\alpha} \| \d {\Gamma} ( b_{\infty }^* b_{\infty } )^{\frac12} \psi_{s}\|^2 \Big )^{\frac12}.
\end{align*}
Since $a_{0} a_{0}^*$ and $ a_{\infty }  a_{\infty }^*$ are of the form $\rho_{1}^* \chi_{b_\e= c t^\al} \rho_{1}$, the weighted minimal velocity estimate \eqref{mve1'_3} implies
\begin{equation*}
\int_1^{\infty} ds \, s^{-\alpha} \| \widehat{\d\Gamma} ( c_{\#1} c_{\#1}^* )^{\frac12} \hat\phi_s\|^2 \lesssim   \| \hat\phi_0 \|^2 ,
\end{equation*}
where $\widehat{\d\Gamma} ( c_{\#1} c_{\#1}^* )^{\frac12}$ stands for $\d \Gamma(  a_{0 }  a_{0 }^* )^{\frac12} \otimes \one$ or $\one \otimes \d \Gamma(  a_{\infty }  a_{\infty }^* )^{\frac12}$. Likewise, since $b_{0 }^* b_{0 }$ and $b_{\infty }^* b_{\infty }$ are of the form $\rho_{-1}^*\chi_{b_\e= c t^\al} \rho_{-1}$, the weighted minimal velocity estimate \eqref{mve1'_2} implies
\begin{equation*}
\int_1^\infty ds\, s^{-\al} \| \d {\Gamma} ( c_{\#2}^* c_{\#2} )^{\frac12} \psi_{s}\|^2 \lesssim   C( \psi_0 ) ,
\end{equation*}
with $c_{\#2} = b_{0}$ or $b_{\infty }$. The last three relations give
\begin{align}\label{eq:G'0tt'_2}
\sup_{\|\hat\phi_0\| = 1} | \int_t^{t'}ds \, \lan \hat\phi_s, \widetilde G'_0 \psi_{s}\ran| \ra 0,\ \quad t, t' \ra \infty.
\end{align}

Applying likewise Lemma C.2 of Appendix C of \cite{FaSig}, 
 one verifies that $\mathrm{Rem}'_t$ satisfies
 \begin{align*}
|Ê\langle \hat \phi ,\mathrm{Rem}'_{t} \psi \ran | \lesssim \| \hat \phi \| \Big ( & t^{-2\alpha + (1+ \kappa)/2 }  \| \d\Gamma( \omega^{-1} )^{\frac12} \psi \| + t^{-1} \| \d\Gamma( \omega^{-1/2} \chi j_\infty'\chi \omega^{-1/2} )^{\frac12} \psi \| \\
& +t^{-\al}\|\d\G( \omega^{-1/2}Ê\chi^2_{(\frac{|y|}{\bar c t})^2 \ge 1} \omega^{-1/2} )^{\frac12} \psi \| \Big ).
\end{align*}
Using \eqref{npb-unif2}, the weighted minimal velocity estimate \eqref{mve1_2} and the weighted maximal velocity estimate \eqref{maxvel-est_2}, we conclude that
\begin{align}\label{eq:Remtt'_2}
\sup_{\|\hat\phi_0\| = 1} | \int_t^{t'}ds \, \lan \hat\phi_s, \mathrm{Rem}'_s \psi_{s}\ran| \to 0, \ \quad t, t' \ra \infty .
\end{align}
Equations \eqref{eq:G'0tt'_2} and \eqref{eq:Remtt'_2} then imply
\begin{equation}\label{tildeG0-contr_2}
\widetilde A = \|\int_t^{t'}ds \, f(\hat H)e^{i\hat{H} s} G_0 \psi_{s}\| \ra 0,\ \quad t, t' \ra \infty.
\end{equation}

The estimate of $ G_1$ is the same as in the proof of Theorem 5.1 of \cite{FaSig}, which shows that $\widetilde W(t)$, and hence $W(t)$, are strong Cauchy sequences. Thus the limit $W_+$ exists. $\Box$

\end{document}